\documentclass[12pt]{article}

\usepackage{amsmath}
\usepackage{amsthm}
\usepackage{amsfonts}
\usepackage{amssymb}
\usepackage{amsmath}
\newtheorem{thm}{Theorem}

\newtheorem{lem}{Lemma}
\newtheorem{proposition}{Proposition}

\newcommand{\F}{\mathbb{F}}

\usepackage[colorlinks]{hyperref}
\hypersetup{pdffitwindow=true,linkcolor=blue,citecolor=red,
urlcolor=black,menucolor=black}

\title{The weight spectrum of two families of Reed-Muller codes }

\author{Claude Carlet,\thanks{The research of the author is partly supported by the Norwegian Research Council.} \\ \small Universities of Paris 8, France and Bergen, Norway. \\ \small {\it E-mail}: {\tt  claude.carlet@gmail.com}\\\ \and Patrick Sol\'e,\\ \small  I2M (CNRS, Aix-Marseille University, Centrale Marseille), Marseilles, France,\\ \small {\it E-mail}: {\tt  sole@enst.fr}\
}
\date{}
\begin{document}
\maketitle
\begin{abstract}
We determine the weight spectra of the Reed-Muller codes $RM(m-3,m)$ for $m\ge 6$ and  $RM(m-4,m)$ for  $m\ge 8$. The technique used is induction on $m$, using that the sum of two weights in $RM(r-1,m-1)$ is a weight in $RM(r,m)$, and using the characterization by Kasami and Tokura of the weights in $RM(r,m)$ that lie between its minimum distance $2^{m-r}$ and the double of this minimum distance. We also derive the weights of $RM(3,8),\,RM(4,9),$ by the same technique. We conclude with a conjecture on the weights of $RM(m-c,m)$, where $c$ is fixed and $m$ is large enough.
%We also give the weights of $RM(3,7),$...
\end{abstract}

\section{Introduction}
For classical Coding Theory notation and terminology, see the next Section. Determining the Hamming weights in Reed-Muller codes has been considered an important research topic for more than half a century
\cite[Chapt. 15]{CC-Sloane}. The weights of the Reed-Muller codes of length $2^m$ and orders $0,1,2,m-2,m-1,m$ are known (as well as their weight distributions). These weights equal $0, 2^m$ for the order 0, with additionally $2^{m-1}$ for the order 1, and $2^{m-1}\pm 2^i$ where $\frac m2\leq i\leq m$ for the order 2, see e.g. \cite{CC-Sloane}. The weights in $RM(m,m)$ are all integers between 0 and $2^m$ since $RM(m,m)=\mathbb F_2^{2^m}$. The weights in $RM(m-1,m)$ are all even integers between 0 and $2^m$ and those in $RM(m-2,m)$ (the extended Hamming code) are all even integers between 0 and $2^m$ except $2$ and $2^m-2$: these two results (as well as the weight distributions) are directly deduced from the Mac Williams identity (see e.g. \cite[Chapt. 5]{CC-Sloane} or \cite{Book}), which says that the weight distribution of a linear code can be obtained as a function of the weight distribution of its dual; the duals of $RM(m-1,m)$ and $RM(m-2,m)$, namely $RM(0,m)$ and $RM(1,m)$, being such that $A_0=A_{2^m}=1$, and for the latter, $A_{2^{m-1}}=2^{m+1}-2$, and all the other coefficients being zero, the weight distributions of $RM(m-2,m)$ and $RM(m-1,m)$ are easily deduced. Another method for determining the weight distribution of $RM(m-2,m)$ is given in \cite{SLNS}, by induction on $m$; we do not give it here since it is more complex than using the MacWilliams transform, but the idea of an induction will be central for the other orders that we shall address. Indeed, the weight distribution of $RM(2,m)$ being already rather complex, because its coefficient $A_{2^{m-1}}$ has no known compact expression, it seems impossible to deduce the weight distribution of $RM(m-3,m)$ from that of its dual, and this is still more true for $RM(m-4,m)$.\\
All the low Hamming weights are known in all Reed-Muller codes: Berlekamp and Sloane \cite{CC-berl-sl} (see the Addendum in this paper) and Kasami and Tokura \cite{CC-Kasami-Tokura} have shown that, for $r\geq 2$, the only
Hamming weights in $RM(r,m)$ occurring in the range $[2^{m-r};2^{m-r+1}[$ 
are of the form~$2^{m-r+1}-2^i$ for some $i$; and the latter have completely
characterized the codewords: the corresponding functions are affinely equivalent to
$$
\begin{cases}
 x_1\cdots x_{r-2}(x_{r-1}x_r\oplus x_{r+1}x_{r+2}\oplus \dots \oplus x_{r+2l-3}x_{r+2l-2}), & \text{for }2\leq 2l\leq m-r+2,\\
x_1\cdots x_{r-l}(x_{r-l+1}\cdots x_r\oplus x_{r+1}\cdots x_{r+l}), &
 \text{for } 3\leq l\leq \min(r,m-r).
\end{cases}
$$

 The functions whose Hamming weights are strictly less than 2.5 times the minimum distance $d=2^{m-r}$ have later been studied in \cite{CC-kasami76:_reed_muller}.\\
Those possible weights of the codewords in the Reed-Muller codes of orders $3,\dots ,m-4$ whose values lie between $2.5\, d$ and $2^m-2.5\, d$ are unknown (but when $m=2r+1$, they are known in some cases by using invariant theory, because the code is then self-dual, see~\cite{CC-Sloane,CC-chhbk}). 

In this note, we completely determine the weights in $RM(m-3,m),$ and $RM(m-4,m)$. Note that our work is on a  different track than works such as \cite{MO} in which the authors consider some weights in some Reed-Muller codes and look for all the functions having these weights, whereas we are looking for the weights (all of them) in some Reed-Muller codes and we do not try to find the functions having these weights. We first observe that the work of \cite[Theorem 16 (2)]{SLNS} can be made much more precise: this reference provides some weights in $RM(m-3,m)$ and we show that these weights are in fact all weights thanks to the result of Kasami-Tokura \cite{CC-Kasami-Tokura}. We extend then the method to the codes $RM(m-4,m).$

The material is arranged as follows. The next section recalls some basic definitions and notions needed to understand the rest of the paper. Sections \ref{sec3}, and \ref{sec4} study the weight spectra of $RM(m-3,m)$, and $RM(m-4,m)$ , respectively. Section \ref{sec5} describes some numerical examples. Section \ref{sec6}
 concludes the article.
%%%%%%%%%%%%%%%%%%%%%%%%%%%%%%%%%%%%%%%%%%%%%
\section{Preliminaries}
The {\em Hamming weight} (in brief, the weight) of an element $x=(x_1,\dots ,x_n) \in \F_2^n$ is the number of indices $i$ such that $x_i \neq 0.$
A binary linear code of {\em length} $n$ is an $\F_2$-subspace of $\F_2^n.$ Its {\em dimension} is its dimension as an $\F_2$-vector space.
Its {\em minimum distance} (in brief, distance) is the minimum Hamming weight of a nonzero codeword.

The {\em extension} $\widehat{C}$ of a linear code $C$ of length $n$ is the linear code of length $n+1$ such that each codeword $(c_0,\dots,c_n) \in \widehat{C}$ satisfies both following conditions 
\begin{enumerate}
 \item $(c_0,\dots,c_{n-1}) \in {C},$
 \item $\sum\limits_{i=0}^{n+1} c_i=0.$
\end{enumerate}

The weights of a code are the Hamming weights of all its codewords. The set of distinct weights (including the zero weight) is called the {\em weight spectrum}\footnote{Contrary to the spectra used for instance in Boolean function theory, the weight spectrum in coding theory does not include the indication of the multiplicities of the weights.}. The list of the numbers, classically denoted by $A_i$, of codewords of weight $i$ for $i$ ranging from 0 to  $n$ is called the {\em  weight distribution} of the code. The Magma notation for this quantity is the list of pairs $<i, A_i>$ where $A_i \neq 0.$

The {\em Reed Muller} codes are a family of binary linear codes of length
$n=2^m$. Given the order $r\in \{0, \dots,m\}$ of such a code (usually denoted by $RM(r,m)$), the  dimension equals $\sum\limits_{i=0}^r {m \choose i}$ and the minimum distance equals $2^{m-r}$.
An explicit definition in terms of Boolean functions is as follows.
Let $B_m$ denote the vector space of polynomials in $m$ variables with coefficients in $\F_2$, that is, of elements of $\F_2[x_1,x_2,\dots,x_m]$, in which the exponent of each variable $x_i$ in each monomial equals 0 or 1. Write $$\F_2^n=\{P_1,P_2,\dots,P_n\}.$$
Let $ev$ denote the evaluation map from $B_m$ to $\F_2^n$ by the rule
$$ev(f)=(f(P_1),\dots,f(P_n)).$$
With this notation we define the Reed-Muller code of order $r$ by

$$RM(r,m)=\{ev(f) \mid f \in B_m \,\& \deg(f)\le r\},$$where $\deg(f)$ is the global degree of the multivariate polynomial $f$ (called the algebraic degree of the Boolean function that $f$ represents).

We will use repeatedly the following Lemma, which is the case $q=2$ of \cite[Cor. 2]{SLNS}.

{\lem For all pairs of integers $(r,m)$ with $0 \le r \le m$, the weight spectrum of $RM(r+1,m+1)$ includes as a subset $S+S,$ where
$S$ is the weight spectrum of $RM(r,m),$ and
$S + S = \{s_1+s_2 \mid s_1, s_2 \in S\}.$
}\\

\noindent The proof of this lemma is particularly simple in our framework: we know (see \cite{CC-Sloane}) that all codewords in $RM(r+1,m+1)$ are obtained as the concatenation of a codeword $u$ in $RM(r+1,m)$ with a word of the form $u+v$, where $v$ is a codeword in $RM(r,m)$ (this is called the ``(u,u+v) construction"). If $u$ is itself taken in the subcode $RM(r,m)$ of $RM(r+1,m)$, then $u$ and $u+v$ are any codewords in this linear subcode and the result follows.\\

We shall refer to the following result as McEliece's congruence.

{\thm The weights in $RM(r,m)$ are multiples of
$2^{\lfloor \frac{m-1}{r}\rfloor}$ \cite{CC-McE}.
This bound is tight in the sense that there is at least a codeword of $ RM(r,m)$ with 
weight $(2t+1)2^{\lfloor \frac{m-1}{r}\rfloor}$ for some integer $t$
\cite{B}.}

We will call the following deep result the Kasami-Tokura bound.
It is a consequence of \cite[Chapt. 15, th. 11]{CC-Sloane}, which also
specifies explicitly $A_w.$

{\thm Let $w$ be a weight of $RM(r,m)$ in the range $2^{m-r}\le w <2^{m-r+1}.$ Let $\alpha=\min(r,m-r),$ and $\beta=\frac{m-r+2}{2}.$ The weight $w$ is of the form $w=2^{m-r+1}-2^{m-r+1-\mu},$ for $\mu$ in the range $1\le \mu\le \max(\alpha,\beta).$ Conversely, for any such $\mu,$ there is a $w$ of that form in the range $2^{m-r}\le w <2^{m-r+1}.$}

We will also require the notion of BCH code of length $n$ and designed distance $d,$ hereby denoted by $BCH(n,d).$ See \cite[Chapt. 9]{CC-Sloane}
for a precise definition.

\section{The weights of the Reed-Muller codes of length $2^m$ and order $m-3$}\label{sec3}

It is shown  in \cite[Theorem 16]{SLNS} by using Magma \cite{M} and by induction on $m$ that all the integers in $\{0,2,4,...,2^m\}\setminus \{2,4,6,10,2^m -2,2^m - 4,2^m -6,2^m -10\}$ are weights in $RM(m-3,m)$ for $m\geq 6$. 
\\
The method used consists in applying Lemma 1, which says that the set of weights in $RM(m-3,m)$ contains $A+A$ where $A$ is the set of weights in $RM(m-4,m-1)$. The authors start then from $RM(3,6)$, whose weights can be obtained by Magma \cite{M}: $\{0,2,4,6,8,...,64\} \setminus \{2,4,6,10,54,58,60,62\}$, and they proceed very simply by induction on $m.$ We shall detail a little their proof (which is slightly informal) and show that the set above covers in fact all weights in $RM(m-3,m)$.
\begin{proposition}\label{p1} For every $m\geq 6$, the weights in $RM(m-3,m)$ are the elements of $\{0,2,4,...,2^m\}\setminus \{2,4,6,10,2^m -10,2^m - 6,2^m -4,2^m -2\}=\{0,8,12+2i,2^m-8,2^m\}$, where $i$ ranges over consecutive integers from 0 to $2^{m-1}-12$.
\end{proposition}
\begin{proof}
The result is correct for $m=6$; assuming it is correct for $m\geq 6$, then denoting the set of weights by $A$, we have that $A$ contains $\{0,8,12,14,\dots,2^m-14,2^m-12,2^m-8,2^m\}$
and therefore  $A+A$ contains 
$\{0,8,12,14,\dots,2^{m+1}-14,2^{m+1}-12,2^{m+1}-8,2^{m+1}\}$, since the case of $0,8,12,14,\dots,2^m-14,2^m-12,2^m-8,2^m$ is covered by adding $0$ to an element of $A$, the case of $2^m,8+2^m,12+2^m,14+2^m,\dots,2^{m+1}-14, 2^{m+1}-12,2^{m+1}-8,2^{m+1}$ is covered by adding $2^m$ to an element of $A$, and the remaining numbers $2^m-10, 2^m-6, 2^m-4, 2^m-2, 2^m+2, 2^m+4, 2^m+6, 2^m+10$ are easily covered as well.  By an induction using Lemma 1, all the numbers in $\{0,8,12+2i,2^m-12,2^m-8,2^m\}$ are then weights in $RM(m-3,m)$.  To complete the proof we show that these are the only possible weights in $RM(m-3,m)$.
\\The minimum distance of $RM(m-3,m)$ being 8, and the code being stable under addition of the all-one vector (the constant Boolean function 1), the numbers $2,4,6,2^m -2,2^m - 4,2^m -6$ cannot be weights, and according to Theorem 2 (\cite{CC-berl-sl,CC-Kasami-Tokura}), the only weights in the integral interval $[8,16)$ are the numbers of the form $16-2^i$ in this interval, and $10$ is not among them. This completes the proof.
\end{proof}

\section{The weights of the Reed-Muller codes of length $2^m$ and order $m-4$}\label{sec4}

Adapting the proof above to the order $m-4$ needs to have the weights of $RM(3,7)$ or better $RM(4,8)$ (since we can observe that starting from $RM(3,7)$ makes us lose some weights). It is impossible to run Magma \cite{M} exhaustively in $RM(3,7)$ (and a fortiori in $RM(4,8)$),  because this code has size $2^{64}$, but 
the weights of $RM(4,8)$ are known from the Online Encyclopedia of Integer Sequences \cite[http://oeis.org/A146976]{OEIS}, and \S 6.1 below includes a direct determination of the weights of $RM(3,7)$. The weights in $RM(4,8)$ are  $$\{0,16,24,28, 30,\dots, 226, 228, 232, 240, 256\},$$ where the dots represent all even integers between 32 and 224. This can be generalized as follows.

\begin{proposition}For every $m\geq 8$, the set of all weights in $RM(m-4,m)$ equals $\{0,16,24,28+2i, 2^m-24, 2^m-16, 2^m\}$, where $i$ ranges over the set of consecutive  integers from 0 to $2^{m-1}-28$.
\end{proposition} 
\begin{proof}
The proof is similar to that of Proposition \ref{p1}: using Lemma 1, we can show that the set of weights contains $\{0,16,24,28+2i, 2^m-28, 2^m-24, 2^m-16, 2^m\}$ and Theorem 2  tells us there is nothing else. Note that by Theorem 1 all weights are even.
\end{proof}

Note that the weight distribution of $RM(5,9)$ is given in Online Encyclopedia of Integer Sequences, see \cite[http://oeis.org/A018897]{OEIS}, and it confirms our result.

\section{Numerical examples}\label{sec5}
In the present section, we show how some known weight spectra can be obtained as well as some new ones.

\subsection{The weights of $RM(3,7)$}

The weights of $RM(2,6)$ are by the results in \cite[Chapter 15]{CC-Sloane} equal to: $$S:=\{0,16,24,28,32,36,40,48,64\}$$(actually, they are given in \cite[OEIS A001726]{OEIS}).

A direct hand calculation or a  computation in Magma \cite{M} yields:
$$S+S=\{ 0, 16, 24, 28, 32, 36, 40, 44, 48, 52, 56, 60, 64, 68, 72, 76, 80, 84, 88, 92,$$ $$
96, 100, 104, 112, 128 \}.$$
By McEliece's congruence (Theorem 1), the weights in $RM(3,7)$ are multiples of $4.$
And by \cite{CC-Kasami-Tokura} the weight $20$ and its complement to $64$ are excluded.\\
Hence, $S+S$ equals the whole weight spectrum  of  $RM(3,7).$

 \subsection{The weights of $RM(3,8)$}
 The present subsection will show why the induction of Proposition 3 must start at 
$m=9.$
The weights in $RM(2,7)$ are  by the results in \cite[Chapter 15]{CC-Sloane} equal to  
$$S:=\{0,32,48,56,64,72,80,96,128\}$$(and this is confirmed by \cite[OEIS A006006 ]{OEIS}).
 A direct hand calculation or a  computation in Magma \cite{M} yields
 $$S+S=\{ 0, 32, 48, 56, 64, 72, 80, 88, 96, 104, 112, 120, 128, 136, 144, 152, 160,$$ $$
168, 176, 184, 192, 200, 208, 224, 256 \}$$

By Theorem 2, %the weight $40$ and its complement to $216$ are excluded, and
this list is complete in the region $\{0,\dots ,64\}$ and in its complement to 256.
However all the elements in $S+S$ are multiple of $8$, while Theorem 1 states that the weights in $RM(3,8)$ are multiples of $4,$ and some weights are not multiple of $8.$
Hence $S+S$ is {\em strictly} included in the  spectrum  of  $RM(3,8),$ since McEliece's congruence is known to be tight \cite{B}.
\\
This is confirmed by the known weight distribution \cite[A146953  ]{OEIS}
\\
{\tt https://isec.ec.okayama-u.ac.jp/home/kusaka/wd/index.html}
\\
which gives in addition the weights: $$100,  108,  116,  124,  132, 140, 148, 
156,  164,$$ that are congruent to $4$ modulo $8.$

These weights can be recovered by taking the intersection of $RM(3,8)$ with the extension of $BCH(255,19).$ The 
weight distribution of this intersection (a code of dimension $26,$ small enough to allow the use of the {\tt WeightDistribution} command of Magma) is:

$$[ <0, 1>, <80, 8>, <88, 56>, <92, 512>, <96, 4939>, <100, 30216>, $$ $$
<104, 159164>,<108, 615184>, <112, 1851060>, <116, 4389152>, $$ $$
<120, 8126540>, <124, 11733960>,<128, 13287280>, <132, 11733960>, $$ $$ <136, 8126540>, <140, 4389152>,<144,
1851060>, <148, 615184>, $$ $$
<152, 159164>, <156, 30216>, <160, 4939>, <164, 512>,<168, 56>, $$$$<176, 8>, <256, 1> ].$$

\noindent {\bf Remark}. We do not provide, properly speaking, a computer-free determination of the weight spectrum of $RM(3,8)$. But we show how it can be derived in a reproducible way. \hfill $\diamond$

\subsection{The weights of $RM(4,9)$}
According to the previous subsection, the weight spectrum of  $RM(3,8)$ equals:
$$S=\{ 0, 32, 48, 56, 64, 68, 72, 76, 80, 84, 88, 92, 96, 100, 104, 108, 112, 116,
120, $$ $$124, 128, 132, 136, 140, 144, 148, 152,
156, 160, 164, 168, 172, 176, 180,
184, 188, $$$$192, 200, 208, 224, 256 \}.$$

A Magma calculation yields
$$S+S=\{0, 32, 48, 56, 64, 68, 72, 76, 80, 84, 88, 92, 96, 100, 104, 108, 112, 116,
$$ $$120, 124, 128, 132, 136, 140, 144, 148, 152, 156, 160, 164, 168, 172, 176, 180,
184, $$ $$ 188, 192, 196, 200, 204, 208, 212,216, 220, 224, 228, 232, 236, 240, 244,
248, 252, $$ $$256, 260, 264, 268, 272, 276, 280, 284, 288, 292, 296, 300, 304, 308,
312, 316, 320, $$ $$ 324, 328, 332, 336, 340,344, 348, 352, 356, 360, 364, 368, 372,
376, 380, 384, 388, $$ $$392, 396, 400, 404, 408, 412, 416, 420, 424, 428, 432, 436,
440, 444, 448, 456, 464, $$$$480, 512
\}.$$
By McEliece's congruence \cite{CC-McE} the weights in $RM(4,9)$ are multiples of $4.$ By Theorem 2 the weights $36,40,44,52$ and their complements to $512$ are excluded, since 
$$64-36=28,\,64-40=24,\,64-44=20,\,64-52=12, $$
which are not powers of $2.$ However, still by Theorem 2 
the integer  $60=64-4=2^6-2^{6-4}$ is a weight of $RM(4,9).$

Note that this integer does not appear in the spectrum of
$RM(3,8)$ since then, in the notation of Theorem 2, $\mu \le \lfloor \max(3,\frac{7}{2})\rfloor=3,$ when $\mu = 4$ for the weight $60$ of $RM(4,9).$ This shows that the statement of Proposition 3 can only be valid for $m\ge 9.$

Hence $\{60\} \cup S+S$ is the whole spectrum  of  $RM(4,9).$

\subsection{The weights of $RM(5,10)$}

We summarize the partial knowledge we have of the spectrum of $RM(5,10).$
\begin{proposition} The weight spectrum of $RM(5,10)$ contains

\begin{enumerate}
\item the integer $62$ and its complement to $1024,$ namely $962$
 \item all the even integers in the range $[448,576]$
 \item the set $$\{ 0, 32, 48, 56, 60, 64, 68, 72, 76, 80, 84, 88, 92, 96, 100, 104, 108, 112,$$ $$
116, 120, 124, 128, 132, 136, 140, 144, 148, 152, 156, 160, 164, 168, 172, 176,$$ $$
180, 184, 188, 192, 196, 200, 204, 208, 212, 216, 220, 224, 228, 232, 236, 240,$$ $$
244, 248, 252, 256, 260, 264, 268, 272, 276, 280, 284, 288, 292, 296, 300, 304,$$ $$
308, 312, 316, 320, 324, 328, 332, 336, 340, 344, 348, 352, 356, 360, 364, 368,$$ $$
372, 376, 380, 384, 388, 392, 396, 400, 404, 408, 412, 416, 420, 424, 428, 432,$$ $$
436, 440, 444, 448, 452, 456, 460, 464, 468, 472, 476, 480, 484, 488, 492, 496,$$ $$
500, 504, 508, 512, 516, 520, 524, 528, 532, 536, 540, 544, 548, 552, 556, 560,$$ $$
564, 568, 572, 576, 580, 584, 588, 592, 596, 600, 604, 608, 612, 616, 620, 624,$$ $$
628, 632, 636, 640, 644, 648, 652, 656, 660, 664, 668, 672, 676, 680, 684, 688,$$ $$
692, 696, 700, 704, 708, 712, 716, 720, 724, 728, 732, 736, 740, 744, 748, 752,$$ $$
756, 760, 764, 768, 772, 776, 780, 784, 788, 792, 796, 800, 804, 808, 812, 816,$$ $$
820, 824, 828, 832, 836, 840, 844, 848, 852, 856, 860, 864, 868, 872, 876, 880,$$ $$
884, 888, 892, 896, 900, 904, 908, 912, 916, 920, 924, 928, 932, 936, 940, 944,$$ $$
948, 952, 956, 960, 964, 968, 976, 992, 1024 \}$$
 
\end{enumerate}

\end{proposition}

\begin{proof}
\noindent
\begin{enumerate}
\item By considering the special Boolean function $x_1x_2x_3x_4x_5 \oplus x_6x_7x_8x_9x_{10}.$

 \item Follows by computing in Magma the weight distribution of
 the intersection between the extension of BCH(1023,157) and
 RM(5,10).
 \item By Lemma 1 applied to the weight spectrum of $RM(4,9)$
 computed in the previous subsection.
\end{enumerate}

\end{proof}

\section{Conclusion and open problems}\label{sec6}

In this note, we have derived the weight spectra of two infinite families of Reed-Muller codes. For further generalization, the base point of the recurrence might not be amenable to exact enumeration by computer and theoretical work will be necessary.
The weight spectra of $RM(4,10)$ and $RM(5,10)$ are not known. (The weight spectra of $RM(3,9)$, which is given in {\tt http://oeis.org/A018895} and of $RM(4,9)$, which is determined in \S 5.2, are not enough for deriving those of $RM(4,10)$ and $RM(5,10)$ by the techniques in this note). This knowledge
is  necessary for starting an induction to determine the spectrum of $RM(m-5,m).$,

Based on the results of Proposition 1 and Proposition 2, it is natural
to conjecture that the weight spectrum of $RM(m-c,m)$  contains all the weights between the minimum distance $2^c$ and its complement to the length $2^m$, that are authorized by McEliece's congruence and Kasami-Tokura's result. Since $\lfloor \frac{m-1}{m-c}\rfloor=1$ for $m> 2c-1,$
we can formulate the following:\\ \ \

\noindent {\bf Conjecture}: Let $c$ be any positive integer. Then for $m> 2c-1$, the weight spectrum of $RM(m-c,m)$ is of the form:
$$\{0\}\cup A \cup B \cup C \cup \overline{B}\cup  \overline{A}\cup \{2^m\},$$

where: \begin{itemize}
       \item $A \subseteq [2^c, 2^{c+1}],$ is given by Kasami and Tokura \cite{CC-Kasami-Tokura},
       \item $B\subseteq [2^{c+1},2^{c+1}+2^c],$ is given by Kasami, Tokura, and Azumi in \cite[Page 392 and foll.]{CC-kasami76:_reed_muller},
        \item $C\subseteq [2^{c+1}+2^c,2^m-2^{c+1}-2^c],$ consists of consecutive even integers,
       \item $\overline{A}$ stands for the complement to $2^m$ of $A$, and $\overline{B}$ stands for the complement to $2^m$ of $B.$
      \end{itemize}

This conjecture is verified for $c=1,2,3,4$. The first open case is $c=5$. \\ \ \

In view of the ternary and quinary analogues of Theorem 15 of \cite{SLNS}, namely Theorems 17 and 18 of \cite{SLNS}, it is natural to ask for analogues of our results for Generalized Reed Muller codes of
characteristics 3 and 5. However, an exact analogue of Theorem 2 in that context does not seem to be available from the literature.\\

{\bf Acknowledgement}. We thank the anonymous Reviewers for their useful comments;  one of these provided the example of Boolean function mentioned in the proof of Proposition 3.

\end{document}